\documentclass[10pt,a4paper,reqno]{amsart}
\usepackage{amsfonts,amsthm,latexsym,amsmath,amssymb,amscd,amsmath, epsf}
\usepackage{graphicx}

\newtheorem{theo+}           {Theorem}
\newtheorem{prop+}           {Proposition}
\newtheorem{coro+}           {Corollary}
\newtheorem{lemm+}           {Lemma}
\newtheorem{conjecture}      {Conjecture}
\theoremstyle{definition}
\newtheorem{defi+}           {Definition}
\newtheorem{problem}         {Problem}

\theoremstyle{remark}
\newtheorem{rema+}           {Remark}

\newenvironment{lemma}{\begin{lemm+}}{\end{lemm+}}
\newenvironment{remark}{\begin{rema+}}{\end{rema+}}

\newcommand{\al}{\alpha}
\newcommand{\be}{\beta}

\newcommand \De {\Delta}

\newcommand {\Ga} {\Gamma}
\newcommand {\ga} {\gamma}

\newcommand{\la}{\lambda}
\newcommand{\La} {\Lambda}
\newcommand{\Si}{\Sigma}

\newcommand{\bC}{\mathbb C}
\newcommand{\bR}{\mathbb R}

\newcommand{\D}{\mathcal D}
\newcommand{\C}{\mathcal C}

\def\newop#1{\expandafter\def\csname #1\endcsname{\mathop{\rm
#1}\nolimits}}

\newop{slc}
\newop{sign}
\newop{mdeg}

\begin{document}
          \numberwithin{equation}{section}
           
          \title [Algebraic spectrum of  quasi-exactly solvable sextic  oscillator]
          {Asymptotics and monodromy of the  algebraic spectrum of  quasi-exactly solvable sextic   oscillator}

\author[B.~Shapiro]{Boris Shapiro}
\address{Department of Mathematics, Stockholm University, SE-106 91
Stockholm,
         Sweden}
\email{shapiro@math.su.se}

\author[M.~Tater]{Milo\v{s} Tater}
\address{Department of Theoretical Physics, Nuclear Physics Institute, 
Academy of Sciences, 250\,68 \v{R}e\v{z} near Prague, Czech
Republic}
\email{tater@ujf.cas.cz}

\date{\today}

\begin{abstract}  Below we study   theoretically  and numerically the asymptotics of the algebraic part of the spectrum for  the  quasi-exactly solvable sextic potential $\Pi_{m,p,b}(x)=x^6+2bx^4+(b^2-(4m+3))x^2$, its  level crossing points,  and its monodromy   in the complex  plane of  parameter $b$. Here  $m$ is a fixed positive integer.   We also discuss connection between the quasi-exactly solvable sextic and the classical quartic potential.
\end{abstract} 

\maketitle

\section{Introduction}

To the best of our knowledge,  historically  first, and the  most  well-known example of a quasi-exactly solvable potential in quantum mechanics  is the quasi-exactly solvable sextic. It was originally discovered in \cite{SBD} (see also  \cite{Tu}, \cite{TuUsh},\cite{UshB}) and it  is given by: 
$$\Pi_{m,p,b}(x)= x^6+2bx^4+(b^2-(4m+2p+3))x^2,$$
where $m$ is a fixed positive integer, $p\in \{0,1\}$, and $b$ is an arbitrary complex number. In \cite{TuUsh} it was shown that, for any value of $b,$ the Schr\"odinger equation 
\begin{equation}\label{eq:Sch}
T=-\frac{d^2}{dx^2}+\Pi(x)=\la y
\end{equation}
 with the boundary conditions
$$y(\pm \infty)=0$$ 
on $\bR$,  
has $m+1$ eigenfunctions of the form
\begin{equation}\label{eq:ans}
\phi(x)=Q(x)e^{-\frac{x^4}{4}-\frac{bx^2}{2}},
\end{equation} 
where $Q(x)$ is an even (resp. odd) polynomial of degree $2m$ (resp. $2m+1$) for $p=0$ (resp. $p=1$.) The above eigenfunctions as well as  their eigenvalues can be found by a simple algebraic procedure presented below.  The latter  eigenvalues form the so-called {\em algebraic part} of the spectrum of $T$.  For any real value of parameter $b,$ these eigenvalues are real and distinct. A number of  their properties is discussed in e.g., \cite {BD}, \cite {BDM}, \cite {Tu}, \cite {Ush}.  

Let us  briefly recall how, for any value of $b$,  to  describe  the algebraic part of the spectrum explicitly.  Simple calculation shows that if an eigenfunction $\phi(x)$ of the  form \eqref{eq:ans} has an eigenvalue $\lambda,$ then its polynomial factor $Q(x)$  satisfies the differential equation: 
\begin{equation}\label{eq:init}
-Q^{\prime\prime}(x)+2(x^3+bx)Q'(x)-((4m+2p)x^2-b)Q(x)=\lambda Q(x).
\end{equation}
 The differential operator 
$$\mathfrak d=-\frac {d^2}{dx^2}+2(x^3+bx)\frac{d}{dx}-((4m+2p)x^2-b)$$ occurring in the l.h.s. of \eqref{eq:init} preserves  the $(m+1)$-dimensional linear space $V_{ev}$ of all  even polynomials of degree $\le 2m$  for $p=0$. For $p=1,$ it preserves the $(m+1)$-dimensional linear space $V_{odd}$ of all odd polynomials of degree $\le 2m+1$. For $p=0,$ using  $t=x^2$,   we can rewrite \eqref{eq:init}  in the form
\begin{equation}\label{eq:main} 
-4t\frac{d^2Q(t)}{dt^2}+(4t^2+4bt-2)\frac{dQ(t)}{dt}-(4mt-b)Q(t)=\la Q(t),
\end{equation}
which is a special case of the double-confluent Heun equation. 

Thus, the  algebraic part of the spectrum of the Schr\"odinger operator $T$  is simply the spectrum of  the operator $\mathfrak d$ restricted to  $V_{ev}$ for $p=0$ (resp. to $V_{odd}$ for $p=1$). Fixing the usual monomial basis $(1,x^2,x^4,...,x^{2m})$ in $V_{ev}$ and $(x,x^3,...,x^{2m+1})$ in $V_{odd},$ we can explicitly calculate the action of $\mathfrak d$ on the respective space. 

Below we will concentrate on  the case $p=0$. (Case $p=1$ is very similar.) Straight-forward calculation shows that,  for $p=0$, the $(m+1)\times(m+1)$-matrix $M_m(b)$ representing the action of  $\mathfrak d$ in the monomial basis $(1,t,t^2,\dots, t^n)$  of $V_{ev}$ coincides with
\begin{equation}\label{eq:matrix}
M_m(b)=\begin{pmatrix} b&-4m&0&0&0&...\\
                                    -1\cdot 2&5b&4-4m&0&0&...\\
                                    0&-3\cdot 4&9b&8-4m&0&... \\
                                    0&0&-5\cdot 6&13b&12-4m&...\\
                       
                                    \vdots&\vdots&  \vdots&\vdots&\vdots&\vdots\\

\end{pmatrix}.
\end{equation}

\medskip
 In what follows, we mainly study different asymptotic spectral properties of $M_m(b)$. 
 The structure of the  paper is as follows. In \S~2  we study the spectral asymptotics   of the sequence  $\{M_m(b)\}$ (after appropriate scaling)  as well as the properties of the  sequence of eigenpolynomials corresponding to a converging sequence of the latter eigenvalues.  Observe that most of the arguments presented \S~2 are not mathematically rigorous since we are missing proofs of several convergence statements. However our arguments provide a natural and numerically supported  heuristics. 
  In \S~3  we present our numerical results and conjectures about the level crossing points  and the monodromy of the spectrum of $M_m(b)$, when $b$ traverses  closed loops in the complex plane avoiding the level crossing points.  Based on our guesses we also formulate the explicit conjecture about the monodromy of the classical quartic potential studied in a large number of papers starting with \cite {BW}. 
  
  Observe that our ``results" below are similar to that of our previous article \cite{ShTaQu}, where the case of quasi-exactly solvable quartic was considered in some details. Although at present we do not know  how one can prove our guesses rigorously, due to their potential importance for physics and many surprising features, we were  advised  by a number scientists including  Professors B.~Simon and  C.~M.~Bender to make our ``results" available to the mathematics and physics communities.

\medskip
\noindent 
{\bf Acknowledgements.} The first author wants to thank Professors A.~Eremenko and A.~Gabrielov of Purdue University for numerous discussions. The first author is grateful to the Nuclear Physics Institute at \v{R}e\v{z} of the Czech Academy of Sciences  for the hospitality in November 2016. The second author acknowledges the hospitality of the Department of Mathematics, Stockholm university in April 2016.  His research was supported by the Czech Science Foundation (GACR) within the project 14-06818S.

\section{Spectral asymptotics} 

To study the characteristic polynomial of the tridiagonal matrix \eqref{eq:matrix}, we follow the circle of ideas developed in \cite{KvA} and use the characteristic polynomials of its principle minors. Namely, denote by 
$\De_m^{(i)},\; i=1,\dots , m+1$ the determinant of the $i$-th  principal minor of $\la I_{m+1}-M_m(b),$ where $I_{m+1}$ is the identity matrix of size $m+1$. This finite sequence of characteristic polynomials satisfies the recurrence relation
\begin{equation}\label{rec:init}
\De_m^{(i)}=(\la-(4i-3)b)\De_m^{(i-1)}-4(2i-2)(2i-3)(m+2-i)\De_m^{(i-2)},\;i=1,\dots, m+1
\end{equation}
with the initial conditions: $\De_m^{(-1)}=0$ and $\De_m^{(0)}=1$.  Observe that, for any fixed real $b,$  $\{\De_m^{(i)}\}_{i=0}^{m+1}$ is a sequence of  discrete orthogonal polynomials. 

Denoting by  $D_m(\la,b):=\De_m^{(m+1)}(\la,b)$, it was proven in \cite {ShTa1} that, for any fixed $b$, the maximal absolute value of the roots of $\D_m(\la,b)$ grows as $16m^{3/2}/3\sqrt{3} $  and the density of the asymptotic root distribution of the scaled polynomials $\{D_m(m^{3/2}\tilde \la,b)\}$ coincides with that of $\{D_m(m^{3/2}\tilde \la,0)\}$ and is given by the  integral
$$\frac{C}{\pi}\int_0^1\frac{d\tau} {\sqrt {64\tau(\tau-1)^2-C^2x^2}},$$
where $x\in[-C,C]$ and $C=16/3\sqrt{3}$; see also \cite {BDM}. 

Therefore, if we keep $b$ fixed and scale $\la$ as above, we get the same standard  limiting distribution.  Using the corresponding three-term recurrence relation, one can easily guess that an interesting dependence of the sequence of spectra on $b$ implying stabilisation of the spectral distribution might happen if one takes the  sequence of matrices 
\begin{equation}\label{eq2}
\widetilde {M}_m(b)=\begin{pmatrix} \frac{b}{m}&\frac{-4m}{m^{{3}/{2}}}&0&0&0&...\\
                                    \frac{-1\cdot 2}{m^{{3}/{2}}}&\frac{5b}{m}&\frac{4-4m}{m^{{3}/{2}}}&0&0&...\\
                                    0&\frac{-3\cdot 4}{m^{{3}/{2}}}&\frac{9b}{m}&\frac{8-4m}{m^{{3}/{2}}}&0&... \\
                                    0&0&\frac{-5\cdot 6}{m^{{3}/{2}}}&\frac{13b}{m}&\frac{12-4m}{m^{{3}/{2}}}&...\\
                                    \vdots&\vdots&  \vdots&\vdots&\vdots&\vdots\\
                                    \end{pmatrix},
\end{equation}
 which is equivalent to scaling $b_m=b m^{1/2}$ and $\widetilde \la=\la/m^{3/2}$ in  \eqref{eq:matrix}. In  terms of the original equation \eqref{eq:main}, we study  degree $m$ polynomial solutions   of the differential equation
\begin{equation}\label{eq:mod}
-4t\frac {d^2Q(t)}{dt^2}+(4t^2+4bm^{1/2}t-2)\frac{dQ(t)}{dt}-(4mt-bm^{1/2})Q(t)=\widetilde \la m^{3/2}Q(t). 
\end{equation}

The characteristic polynomials $\widetilde \De^{(i)}_m(\widetilde \la,b)$ of the principal minors of \eqref{eq2} satisfy the modified $3$-term recurrence 
\begin{equation}\label{rec:mod}
\widetilde\De^{(i)}_m=\left(\widetilde \la -\frac{4i-3}{m} b\right)\widetilde \De^{(i-1)}_m-\frac{4(2i-2)(2i-3)(m+2-i)}{m^3}\widetilde \De_m^{(i-2)},
\end{equation}
where $ \De^{(-1)}_m=0,\; \widetilde \De^{(0)}_m=1$. See examples in Fig.~\ref{fig0}.  

\begin{figure}
\begin{center}
\includegraphics[scale=0.55]{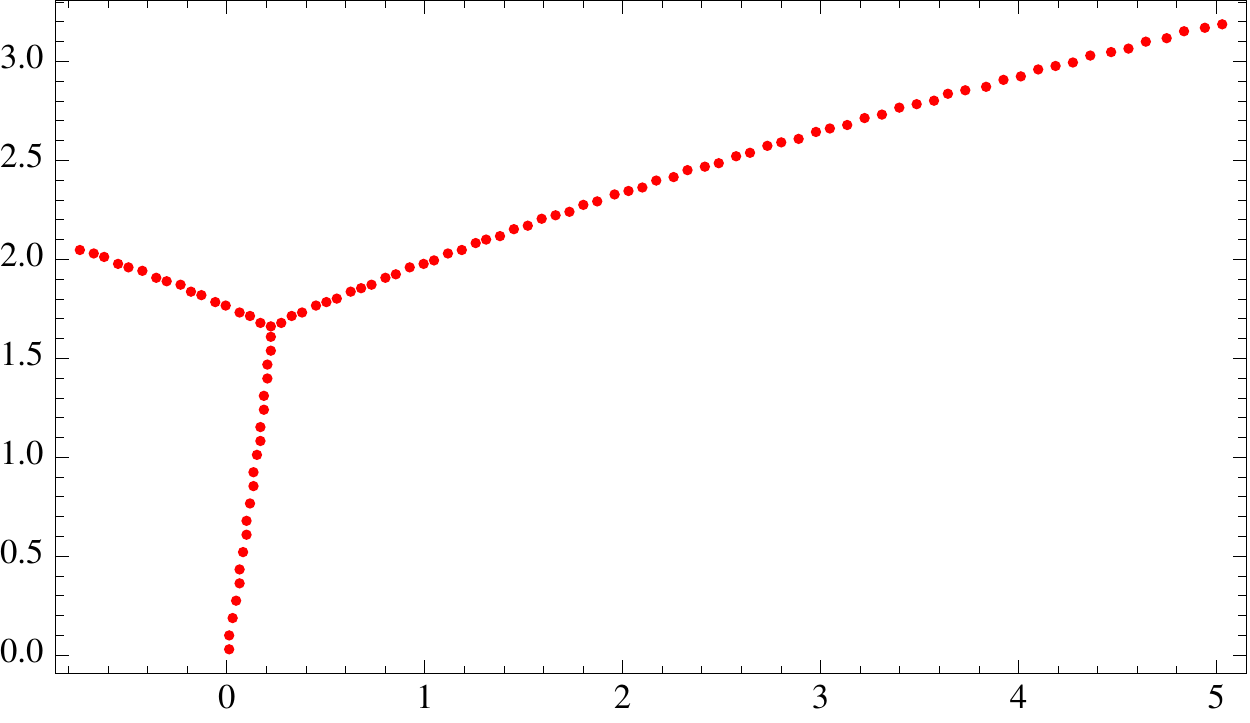} \includegraphics[scale=0.35]{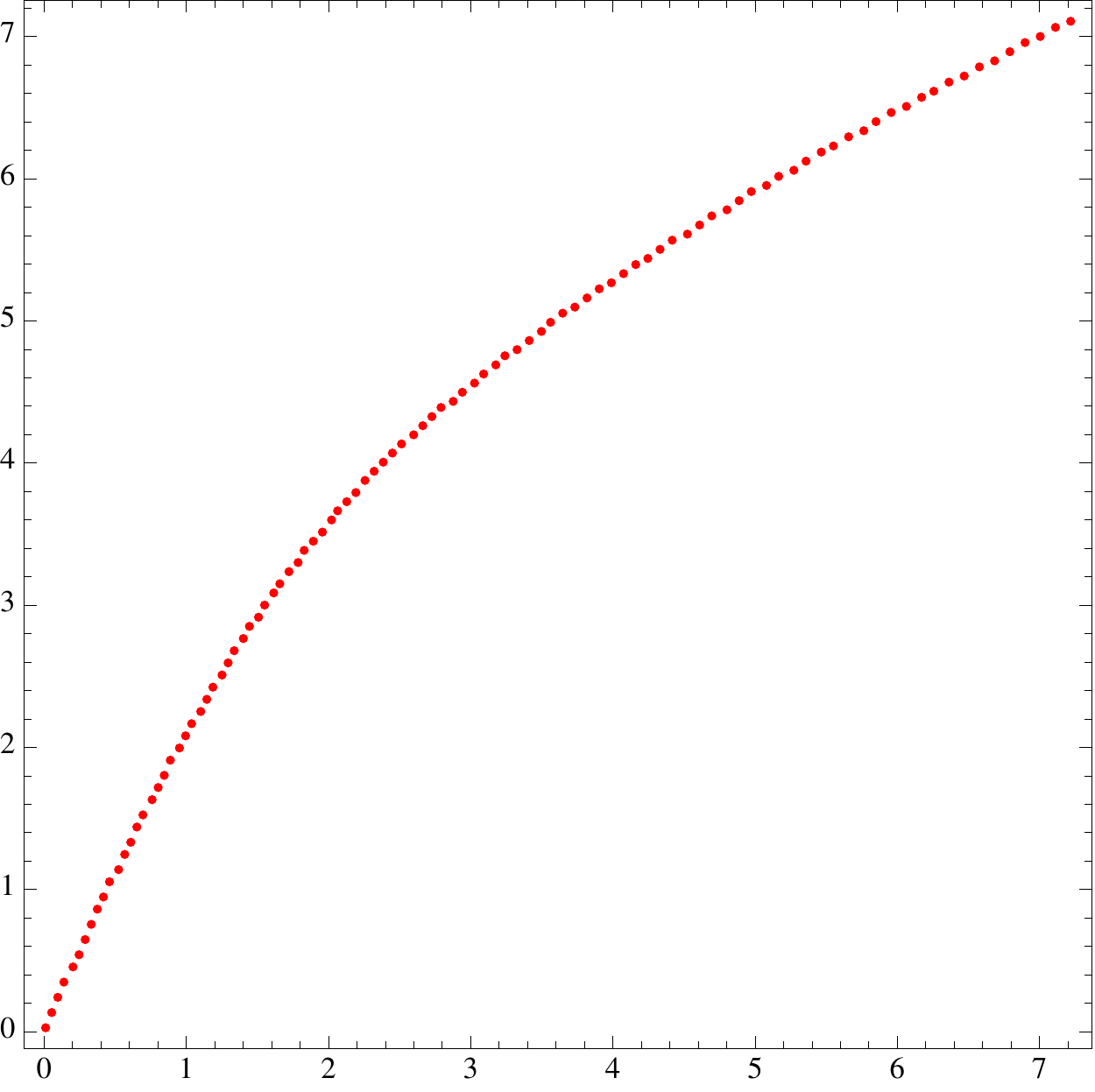}
\end{center}

\vskip 1cm

\caption{The spectra of $\widetilde M_{100}(b)$ for $b=(3/4+I)$ (left) and $b=3/2+2I$ (right).}
\label{fig0}
\end{figure}

Set $\widetilde D_m(\widetilde \la, b):=\widetilde \De_m^{(m+1)}(\widetilde \la, b)$. Below, for any given $b$,  we will study  the asymptotic root-counting measure $\mu_b$ of the polynomial sequence $\{\widetilde D_m(\widetilde \la, b)\}$.  By the main result of \cite{KvA}, the Cauchy transform of $\mu_b$ outside a certain bounded domain in $\bC$ can be calculated  by averaging the Cauchy transforms of polynomial sequences in a $1$-parameter family (depending on parameter $\tau\in [0,1]$) which is obtained from \eqref{rec:mod}  by taking the limit $\frac{i}{m} \to \tau$. In other words, we need to consider the one-parameter family of three-term recurrence relations of the form
 \begin{equation}\label{eq:param}
\De_\tau^{(i)}=(\widetilde \la-4\tau b)\De^{(i-1)}_\tau-4(2\tau)^2(1-\tau)\De_\tau^{(i-2)},\quad\quad \tau\in[0,1].
\end{equation} 
The characteristic equation  of \eqref{eq:param}  is given by 
$$\Psi^2=(\widetilde \la -4\tau b)\Psi-16\tau^2(1-\tau).$$
Its branch points with respect to $\Psi$ (i.e. the values of $\widetilde \la$ for which the latter characteristic equation has a double root with respect to $\Psi$) are determined by the relation 
$$(\widetilde \la-4\tau b)^2=64\tau^2(1-\tau).$$
In other words, 
\begin{equation}\label{endpoints}
\widetilde \la_{1,2}(\tau)=4\tau b\pm8\sqrt{\tau^2(1-\tau)},\quad \tau\in[0,1].
\end{equation}

By \cite{KvA},  in the complement of the domain traversed by the family of the straight segments $[\widetilde \la_{1}(\tau),\widetilde \la_{2}(\tau)]$ in $\bC,$ where $\tau$ runs over the interval $[0,1],$  the Cauchy transform of $\mu_b$ is given by the integral  formula 
\begin{equation}\label{eq:Cauchy}
\C_b(z)=\int_0^1\frac{d\tau} {\sqrt{(z-4\tau b)^2-64\tau^2(1-\tau)}}.
\end{equation}

If $b$ is real, then each interval $[\widetilde \la_{1}(\tau),\widetilde \la_{2}(\tau)]$ is real and one can show  that  
$$\bigcup_{\tau\in[0,1]}[\widetilde \la_{1}(\tau),\widetilde \la_{2}(\tau)]=\left[\frac{2}{27}\left(36b-b^3-\sqrt{(12+b^2)^3}\right),\frac{2}{27}\left(36b-b^3+\sqrt{(12+b^2)^3}\right)\right].$$
 The density of $\mu_b$ on the latter interval can be represented by the integral formula similar to the case $b=0$ given above. 

If $b=u+iv$ with $v\neq 0$, then the family of the above endpoints \eqref{endpoints} traverses the oval of the real rational  cubic $\Ga_b$ with a node at the origin given by the equation
\begin{equation}\label{cubic}\Ga_b: 
\left(x-\frac{uy}{v}\right)^2=\frac{(4v-y)y^2}{v^3},
\end{equation}
where  $x=\text{Re } z$ and $y=\text{Im } z,$ see examples in Fig.~\ref{fig2}. 
Explicit parameterisation of $\Ga_b$ as a rational curve is given by 
$$\begin{cases} x=\frac{2v\nu}{(1-\nu^2)^3} \left(4(1-\nu^2)^2-(2v\nu-u(1-\nu^2))^2\right)\\
                           y=\frac{v}{(1-\nu^2)^2}\left(4(1-\nu^2)^2-(2v\nu-u(1-\nu^2))^2\right).
\end{cases}
$$
(To get the oval, $\nu$ has to run between the roots of the equation $4 (1 - \nu^2)^2 - (2 v \nu - u (1 -
\nu^2))^2 = 0$.) 

\begin{figure}
\begin{center}
\includegraphics[scale=0.5]{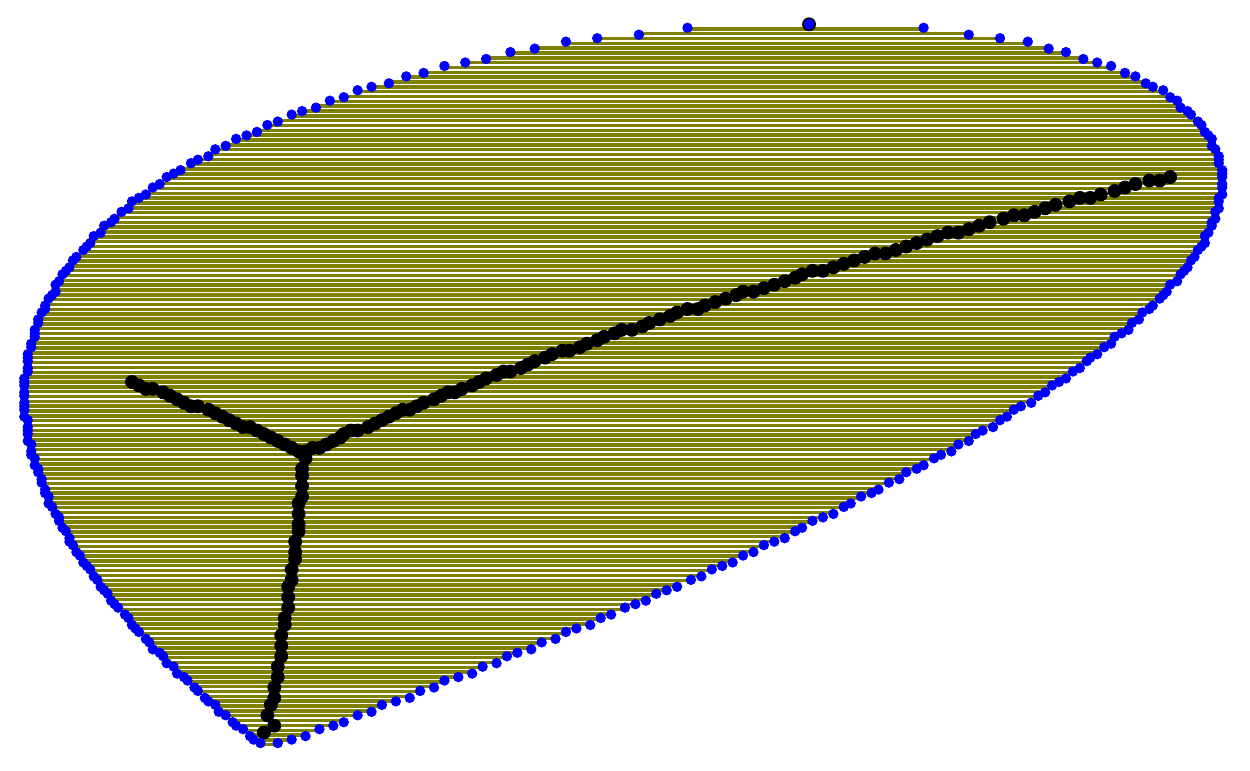} \includegraphics[scale=0.35]{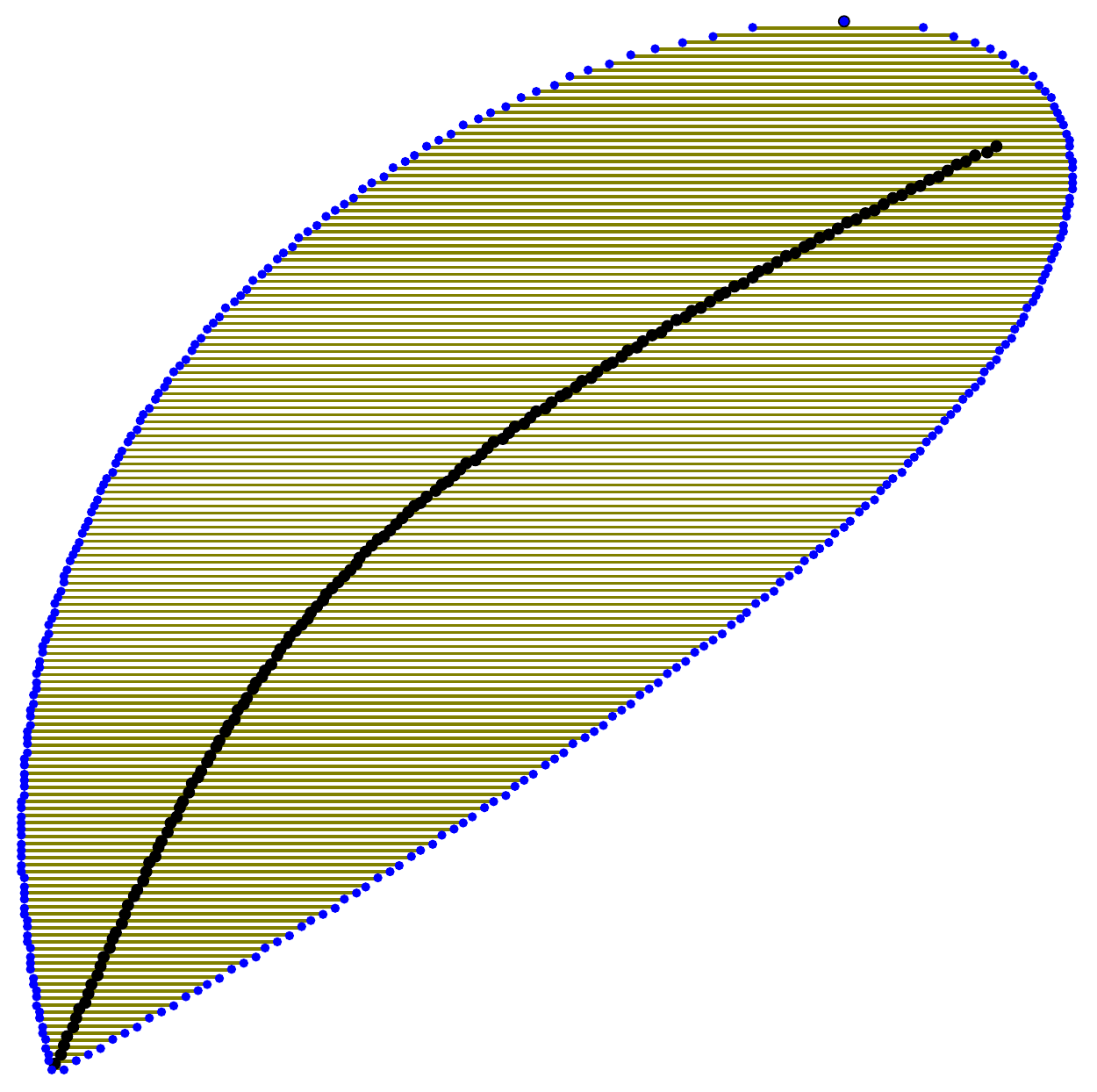}
\end{center}

\vskip 1cm

\caption{The spectra of $\widetilde M_{100}(b)$ for $b=(3/4+I)$ (left) and $b=3/2+2I$ (right) with the corresponding ovals and foci.}
\label{fig2}
\end{figure}

Recall that the general notion of (real) foci of real algebraic curves was developed by J.~Pl\"ucker around 1832   and these foci are the intersections with the real plane of the complex tangent lines to the complexification of the original curve passing through the so-called circular points at infinity. In the homogeneous coordinates $(X,Y,Z)$ of the plane $(x,y)$, these circular points  are given by 
$(1,\pm I, 0)$, see e.g. \cite{Sa}. The easiest way to find the foci of a real plane algebraic curve given by the equation $F(x,y)=0$ is to substitute $y=I x+f$ in $F$ and to calculate the discriminant of the resulting polynomial with respect to $x$. This leads to the polynomial equation in the variable $f$ whose roots are the foci of the original curve, see e.g. \cite{Em}. In the case of the  singular cubic $\Ga_b,$  we get  the equation
\begin{equation}\label{foci}
f^2(27f^2+4f(b^3-36b)-16(4-b^2)^2)=0.
\end{equation}
Equation~\eqref{foci} has a double root at the origin and two more foci $f_{1,2}$ given by
$$f_{1,2}=\frac{2}{27}\left(36b-b^3\pm \sqrt{(12+b^2)^3}\right).$$
Observe that,  in the case of real $b$,  $f_{1,2}$ coincide with the endpoints of the interval $\bigcup_{\tau\in[0,1]}[\widetilde \la_{1}(\tau),\widetilde \la_{2}(\tau)]$ obtained above. 
Our numerical results strongly support the following conjecture.

\begin{conjecture}
 Depending of the value of $b,$ the endpoints of the support of $\mu_b$ are either all three foci of $\Ga_b$ or just two of them always including the focus at the origin, see Fig.~\ref{fig2}.
 \end{conjecture}

In order to characterize the support of $\mu_b$ completely, we suggest the following heuristic argument. Let $\widetilde \la_{j_m,m}$ be an eigenvalue of \eqref{eq:mod}. Abusing our notation,  denote by $p_m(t)$ the eigenpolynomial whose eigenvalue equals $\widetilde \la_{j_m,m}$. It satisfies the differential equation
\begin{equation}\label{eigenp}
-4tp_m^{\prime\prime}+(4t^2+4bm^{1/2}t-2)p^\prime_m-(4mt-bm^{1/2}+\widetilde \la_{j_m,m}m^{3/2})p_m=0.
\end{equation}

Let us  choose a (sub)sequence $\{\widetilde \la_{j_m,m}\}_{m=1}^\infty$ of the eigenvalues of $\widetilde M_m(b)$ (one $j_m$ for each $m$) and assume that it converges to some complex number $\La.$ In other words, $\La=\lim_{m\to \infty}\widetilde \la_{j_m,m}$. Let $\{p_m(t)\}_{m=1}^\infty$ be the sequence of corresponding eigenpolynomials, and $\{\kappa_m\}_{m=1}^\infty$ be the sequence of their root-counting measures.  One can easily observe that the sequence $\{\kappa _m\}_{m=1}^\infty$  does not converge without appropriate scaling. Scaling the variable $t$ in the $m$-th eigenpolynomial as $t=\Theta m^{1/2},$ $\Theta$ being the new scaled time variable, we transform equation~\eqref{eigenp} into
$$
 \frac{-4\Theta}{m^{1/2}}\frac{d^2p_m}{d\Theta^2}+\frac{(4m\Theta^2+4mb\Theta-2)}{m^{1/2}}\frac{dp_m}{d\Theta}-(4m^{3/2}\Theta-bm^{1/2}+\widetilde \la_{j_m,m}m^{3/2})p_m=0.
$$
Dividing the above equation by $m^{3/2}p_m,$ we get 
\begin{equation}\label{eigenpmod}
-4\Theta\frac{ \frac{d^2p_m}{d\Theta^2}}{m^2p_m}+\left(4\Theta^2+4b\Theta-\frac{2}{m}\right)\frac{\frac{dp_m}{d\Theta}}{mp_m}-\left(4\Theta-\frac{b}{m}+\widetilde \la_{j_m,m}\right)=0. 
\end{equation}

Denote by   $\{\widetilde \kappa_m\}_{m=1}^\infty$ the sequence of the root-counting measures of the scaled eigenpolynomials $p_m(m^{1/2}\Theta)$. Assuming that the weak limit $\lim_{m\to\infty}\widetilde \kappa_m$ exists, denote it by $\Omega$. Then,  
$$\lim_{m\to \infty} \frac{\frac{dp_m}{d\Theta}}{mp_m}= \C_\Omega\quad \text{and}\quad   \lim_{m\to \infty} \frac{\frac{d^2p_m}{d^2\Theta}}{m^2p_m}= \C_\Omega^2,$$ 
where the limits are understood as distributions. Thus under the convergence assumption, the Cauchy transform $\C_\Omega$ of the limiting measure $\Omega$ satisfies a.e. in $\bC$ the algebraic equation
\begin{equation}\label{eq:final}
\Theta\C^2_\Omega-\Theta(\Theta+b)\C_\Omega+(\Theta+\Lambda/4)=0,
\end{equation}
where $\Lambda=\lim_{m\to\infty}\widetilde\la_{j_m,m}$. The fact that \eqref{eq:final} admits a solution which is the Cauchy transform of a probability measure  supported on a finite number of compact semi-analytic curves and points imposes strong restriction on the possible values of $\Lambda$. 

We need the following statement, see \cite {BoSh}. 

\begin{lemma}\label{QDiff}
If the Cauchy transform $\C_\nu$ of a probability measure $\nu$ satisfies a.e. in $\bC$ a quadratic equation
$$Q_2(z)\C_\nu^2+Q_1(z)\C_\nu+Q_0(z)=0,$$  
then the support of $\nu$ consists of finitely many semi-analytic curves which are horizontal trajectories of the quadratic differential $\Psi=- \frac{Q_1^2-4Q_2Q_0}{Q_2^2}dz^2$.  In particular, the finite endpoints of the support are either the roots of $Q_1^2-4Q_2Q_0$ or the roots of  $Q_2$. 
\end{lemma}

In our case the corresponding quadratic differential is 
\begin{equation}\label{quadr}
\Psi_{b,\Lambda}=- \frac{\Theta(\Theta+b)^2-4\Theta-\Lambda}{\Theta} d\Theta^2.
\end{equation}

\begin{lemma}\label{ends}

The set of critical $\La$ of the polynomial $P(\Theta)=\Theta(\Theta+b)^2-4\Theta-\La$, (i.e., the set of $\La$ for which $P(\Theta)$ has a double root w.r.t $\Theta$)  coincides with the foci \eqref{foci}. 
\end{lemma}

\begin{proof} Straight-forward calculation.
\end{proof}

 For generic $\Lambda$, the differential $\Psi_{b,\Lambda}$ has three simple zeros and one simple pole at $0$. Since $\Omega$ has bounded support and the Cauchy transform $\C_\Omega$ is univalent in the complement to the support of $\Omega$ (which consists of a finite number of compact curves and points), then $\Psi_{b,\Lambda}$ must  have two  critical trajectories one of which  connects two (simple) zeros and the other connects  the pole at the origin and the remaining  zero. This reasoning motivates the following claim. 

\begin{conjecture}\label{relation} The support of $\mu_b$ coincides with the set of values of $\Lambda$ such that the differential \eqref{quadr} has two critical horizontal trajectories. 
\end{conjecture}


\section{Numerical results on the level crossing points and spectral monodromy of QES-sextic}

\subsection{Level crossing points}  For any fixed $b$, the spectrum of $M_m(b)$ is the zero locus of the bivariate polynomial $D_m(\la,b)$ with respect for the variable $\la$.  For a generic value of $b$, the spectrum of $M_m(b)$ consists of $m+1$ distinct points. By definition, the level crossing  $\Si_m\subset \bC$ of $M_m(b)$ is the set of all values of $b$, for which the spectrum of $M_m(b)$ contains less than $m+1$ distinct points. $\Si_m$ is the zero locus of the discriminant $\D_m(b)$ of $D_m(\la,b)$ determined as 
$$\D_m(b):=\text{Resultant} \left( D_m(\la,b), \frac{\partial D_m ( \la,b) } {\partial \la}, \la \right).$$
One can easily check that $\deg \D_m(b)=m(m+1)$. Observe that $\D_m(b)$ is  a real univariate polynomial without real roots. Thus the set $\Si_m$ consists of $\binom {m+1}{2}$ (not necessarily distinct) complex conjugate pairs of points. 
The level crossing sets $\Si_{10}(b)$ and  $\Si_{20}(b)$ are shown in Fig.~\ref{figromb}.  (Similar picture can be found in \cite{ErGaIr}.) Our experiments in Mathematica for $m\le 25$ show that:  

\medskip
\noindent (i) $\Si_m$   forms a lattice-like pattern whose outer boundary is a curvilinear rombus. The level crossing points in the upper and lower half planes are naturally organized in ``horizontal" rows with $m, m-1, \dots, 1$ points respectively, see Fig.~\ref{figromb}.

\smallskip
\noindent (ii) the sizes of the rhombi grow as $\sqrt{12m}$, see Fig.~\ref{scale};

\smallskip
\noindent (iii) after scaling by $\sqrt{m}$, the sequence of root counting measures of $\Si_m$ converges to a continuous measure $\mu_b$ supported on the curvilinear rhombus $\mathfrak R\subset \bC$ in the complex $b$-plane such that if $b\in \mathfrak R,$ then the support of $\mu_b$ consists of three legs ending at $0$ and both foci $f_1, f_2$, see Fig.\ref{fig2} left. If $b\in \bC\setminus \mathfrak R,$ then the support of $\mu_b$ consists of one leg ending at $0$ and one of the  foci $f_1, f_2$, see Fig.\ref{fig2} right. The boundary of $\mathfrak R$ consists of those $b$ for which one focus lies on the leg connecting $0$ with the other focus, see Fig.~\ref{fig11}. Conjecturally, the four vertices of $\mathfrak R$ are $\pm 2$ and $\pm \sqrt{12}I$.  At $b=\pm 2$ one of the foci coincides with $0$ and at $b=\pm \sqrt{12}I$ the foci coincide with each other. 

\smallskip
\noindent (iii) after appropriate rescaling, the distribution of level crossings near the center of each of curvilinear triangles converge to a regular hexagonal lattice, see Fig.~\ref{fig3}. (Close to the center of the curvilinear triangle in Fig.~\ref{fig3} the picture resembles a hexagonal lattice.)

\begin{figure}

\begin{center}
\includegraphics[scale=0.4]{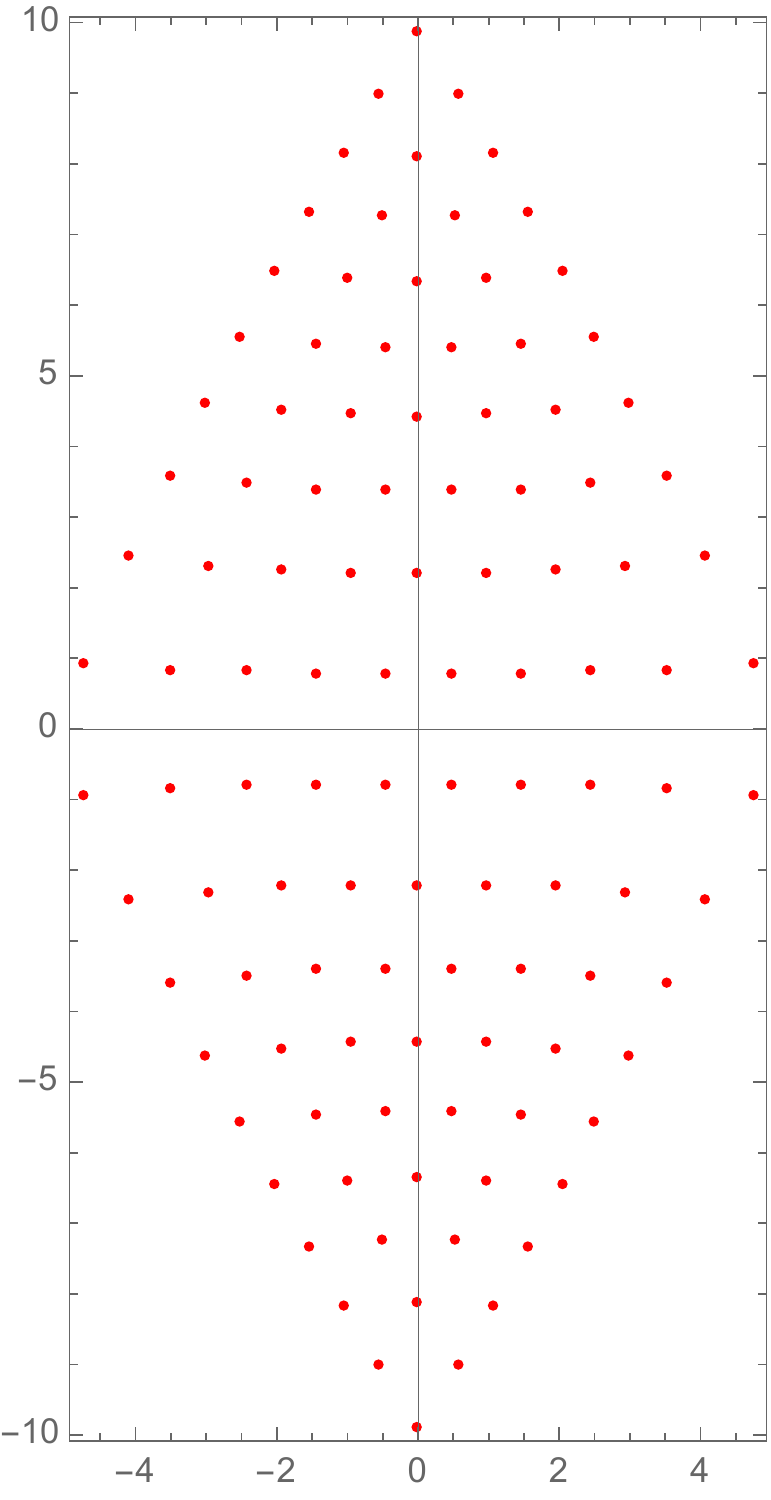}   \hskip 2cm \includegraphics[scale=0.4]{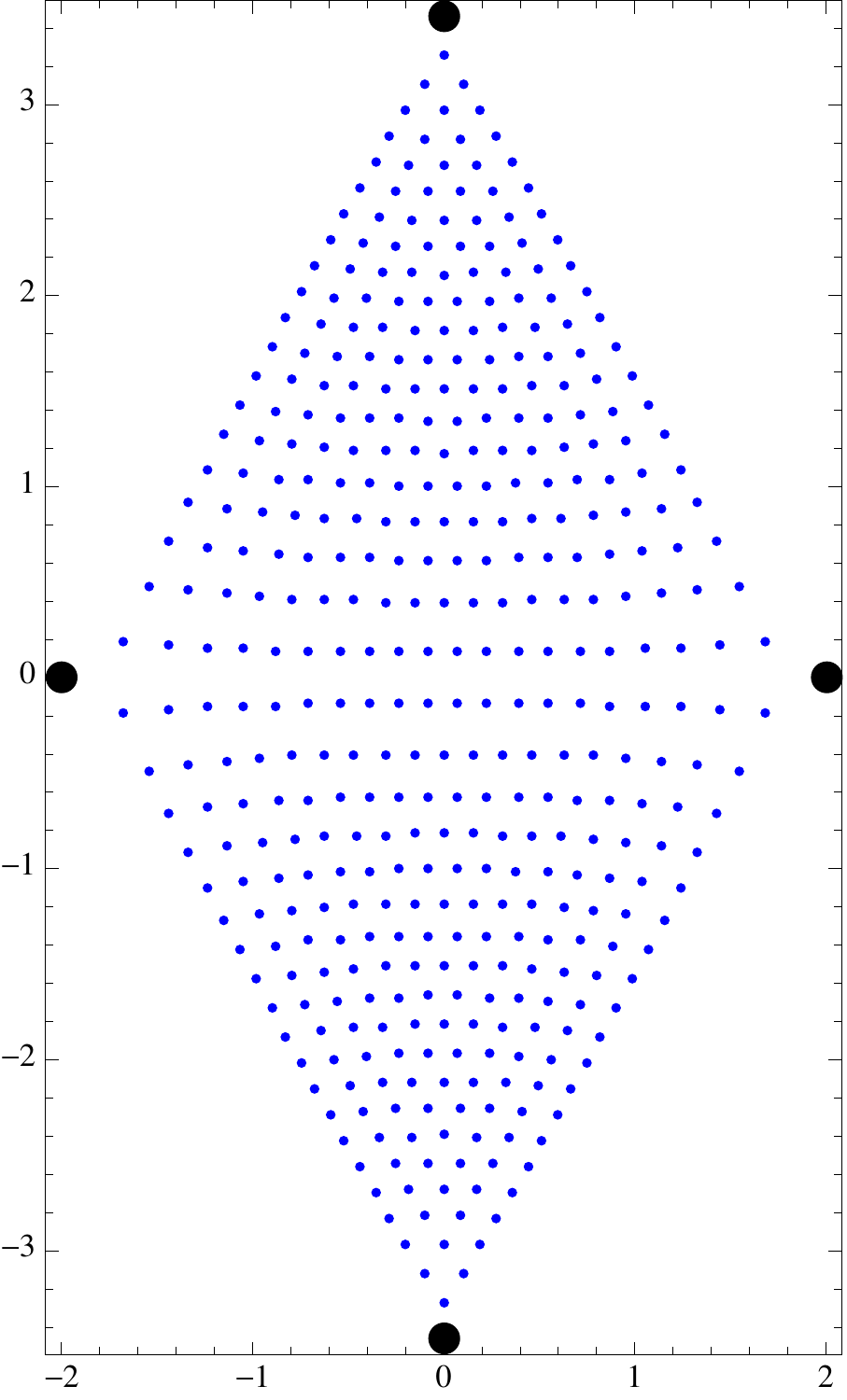} 
\end{center}


\caption{Rombi of the level crossing points of the QES-sextic for $m=10$ and $m=20$ with scaling.}
\label{figromb}
\end{figure}

\begin{figure}

\begin{center}
\includegraphics[scale=0.35]{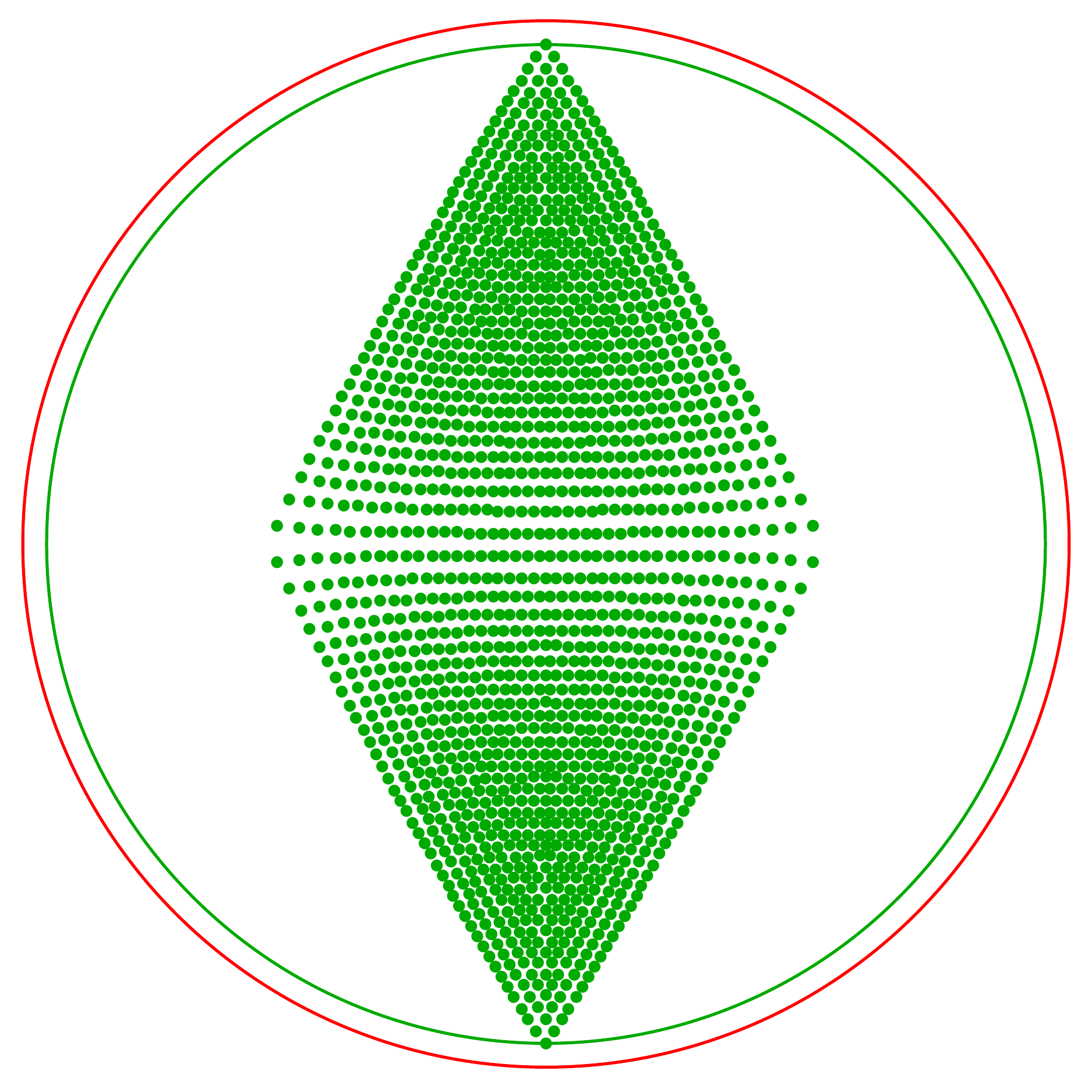}  
\end{center}


\caption{Rombi of the level crossing points of the QES-sextic for $m=41$. The red circle has radius $\sqrt{12\cdot 41}$.}
\label{scale}
\end{figure}

\begin{figure}

\begin{center}
\includegraphics[scale=0.3]{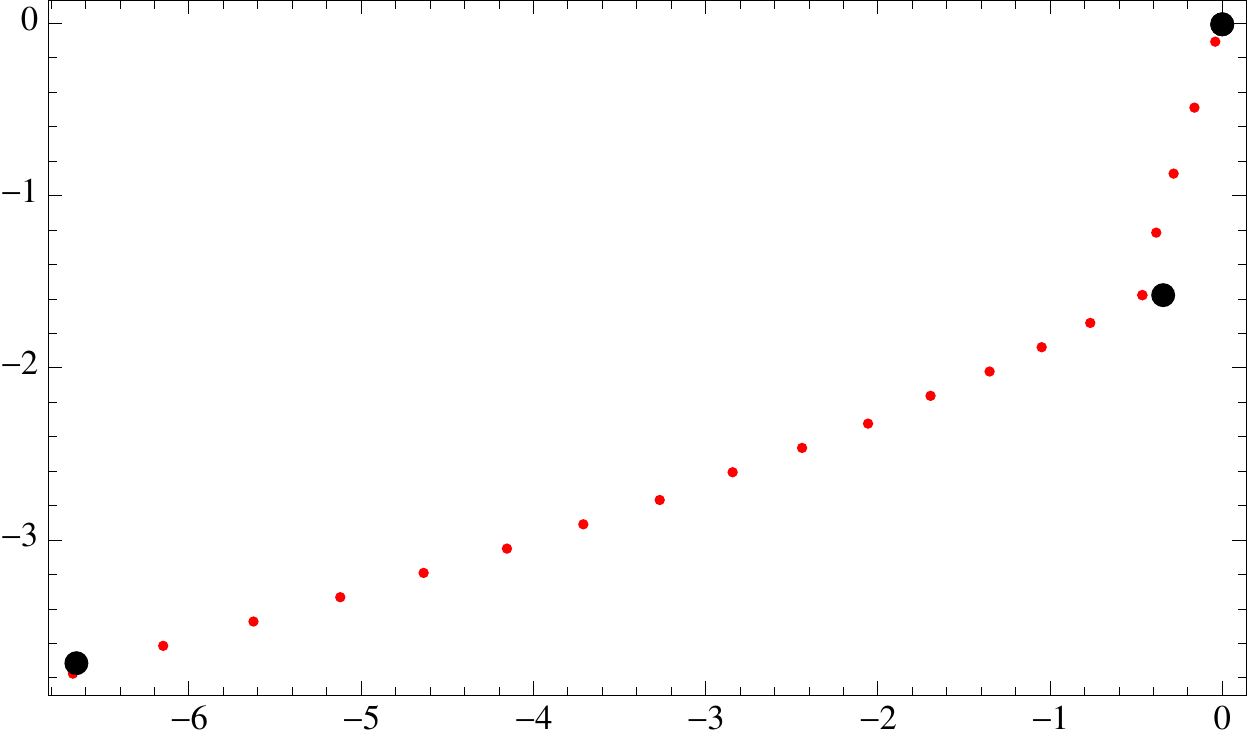} 
\end{center}


\caption{The support of $\mu_b$ for $b$ on boundary of $\mathfrak R$.}
\label{fig11}
\end{figure}

\begin{figure}

\begin{center}
\includegraphics[scale=0.4]{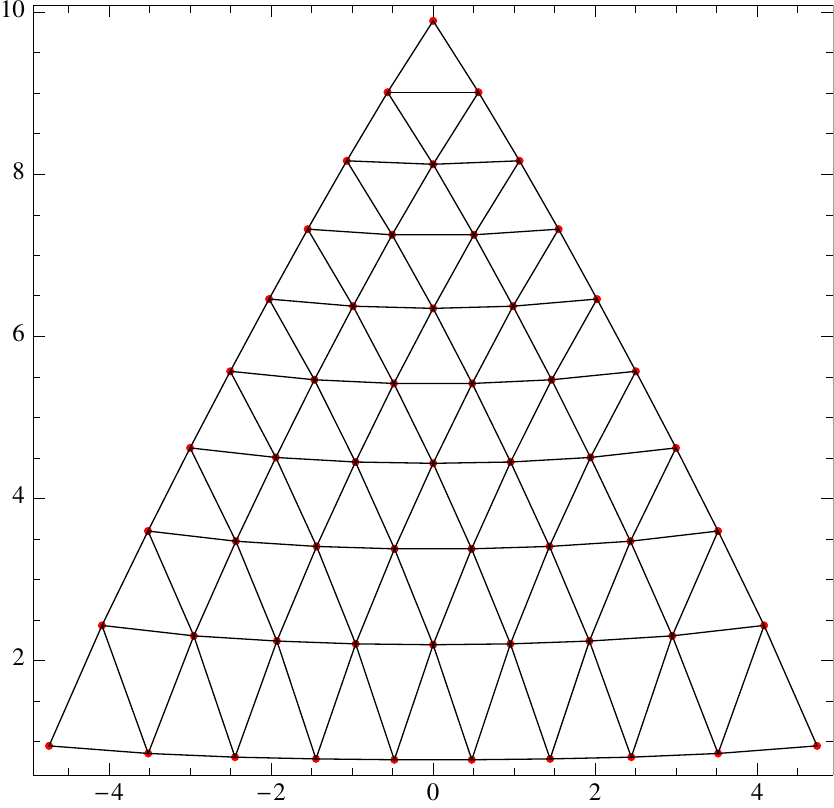}
\end{center}


\caption{Level crossings in the upper half plane for $m=10$ with  nearest neighbours  connected.}  
\label{fig3}
\end{figure}
   
\subsection{Monodromy} Below,  based on our numerical results obtained for $m\le 10$, we present,  for any  positive integer $m$,  an explicit conjecture completely describing the monodromy of the spectrum of $M_m(b)$.  Fig.~\ref{fig4}  shows these numerical results for $m=5$. 

To determine the monodromy operators of the spectrum of $M_m(b)$, one  needs  to choose a system of  (based) loops in $\bC\setminus \Si_m$ such that they generate the fundamental group of the latter space. Since $\bC\setminus \Si_m$ is a wedge of $m(m+1)$ circles, we need to choose $m(m+1)$ loops in our system.  Due to the fact that the spectrum of $M_m(b)$ is real and simple for all real $b$ and that $\Si_m$ consists of complex-conjugate pairs, we suggest the following system of loops. For each level crossing point $b_j=u_j+Iv_j$ in the upper half plane, construct the loop $\ga_j$ which starts at $u_j\in \bR$; goes up almost to $b_j$; traverses counter-clockwise a small circle centered at $b_j$, and returns back to $u_j$ moving vertically down.      
 Observe that, in principle,  such loops can  pass through other level crossing points which is forbidden by definition. But our numerical experiments show that:
 
 \noindent
 (a) such a situation happens only when $u_j=0,$ i.e. for the purely imaginary level crossing points, and 
 
 \noindent
 (b) for the purely imaginary level crossing points, one can make an arbitrary small deformation of $\ga_j$ to avoid collision with other crossing points and the resulting monodromy will be independent of the deformation.

\begin{figure}

\begin{center}
\includegraphics[scale=0.2]{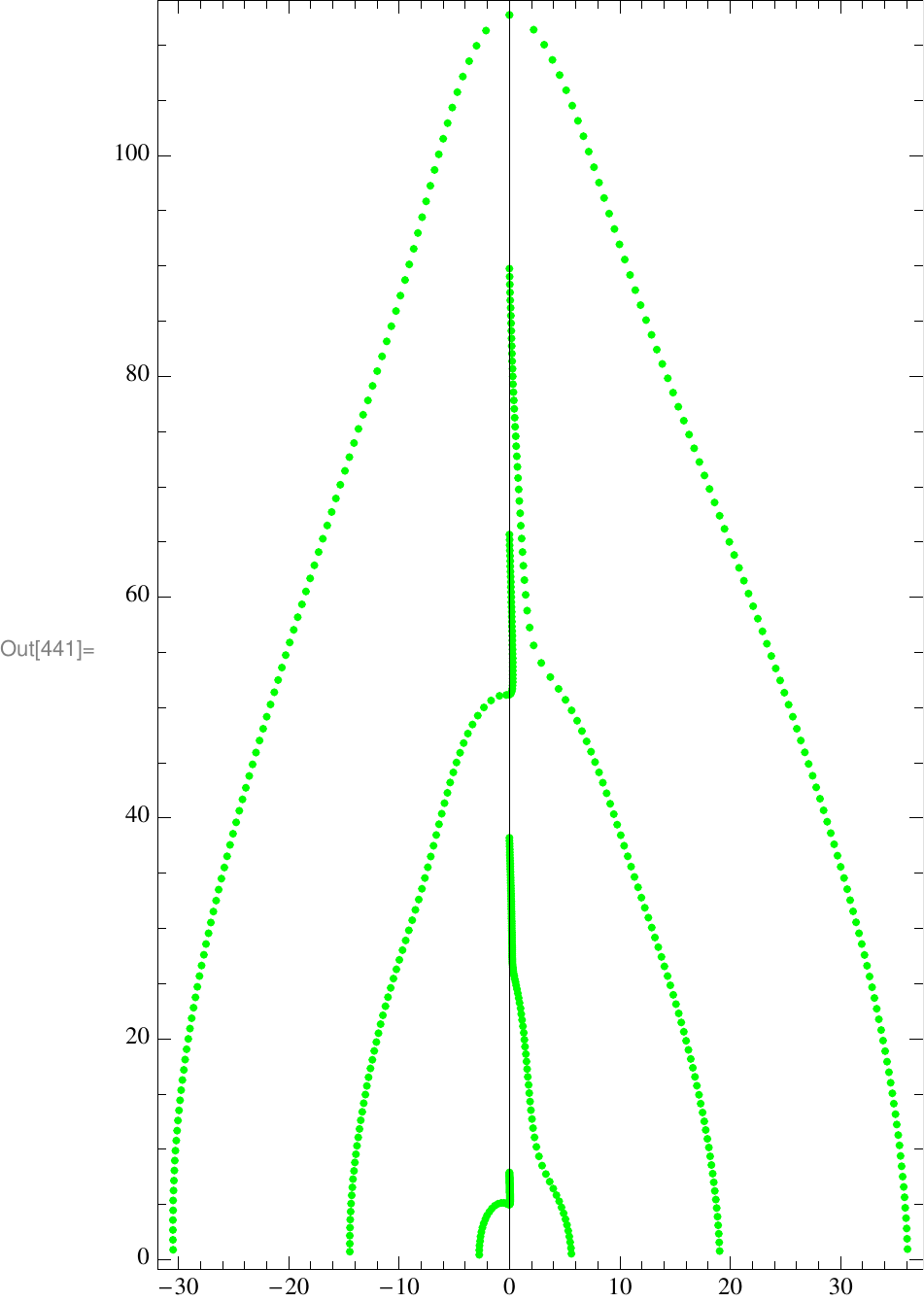}
\end{center}

\vskip 0.5cm
\begin{center}
\includegraphics[scale=0.2]{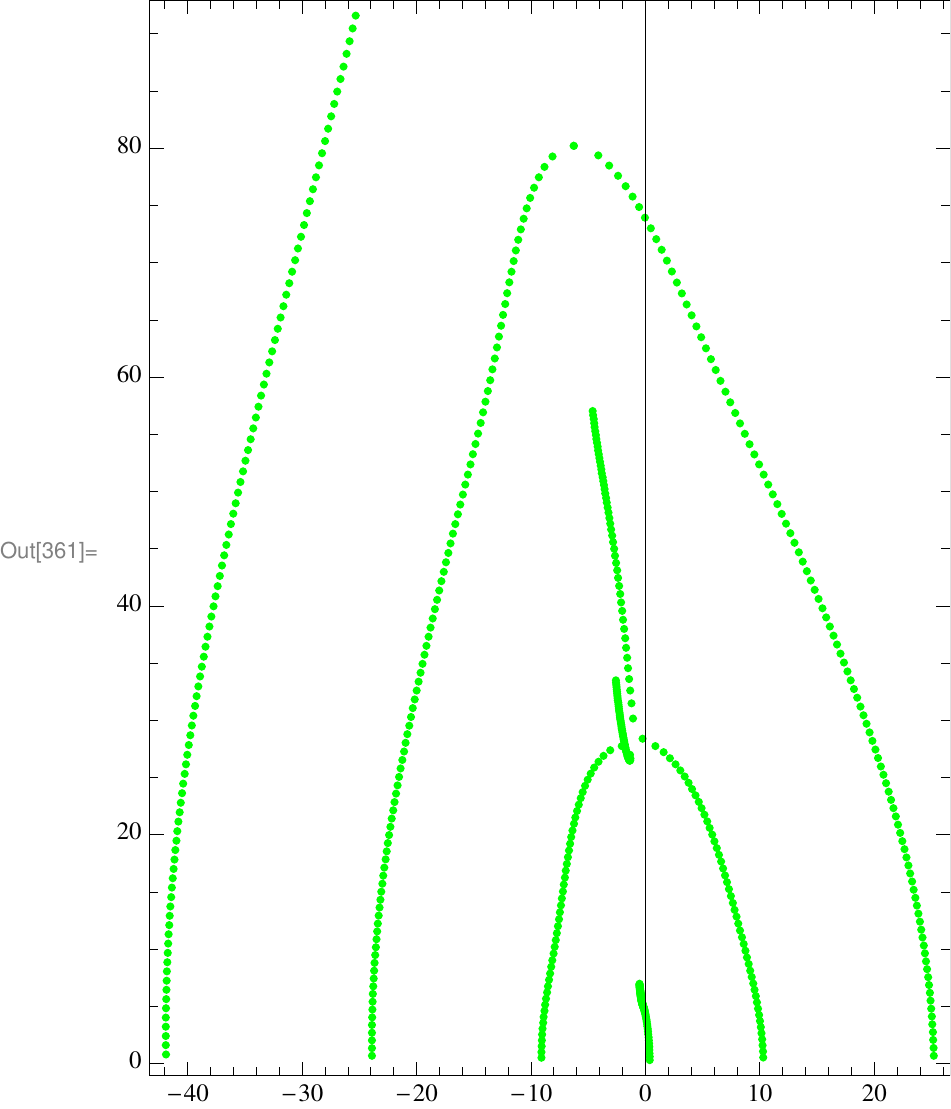}
\includegraphics[scale=0.2]{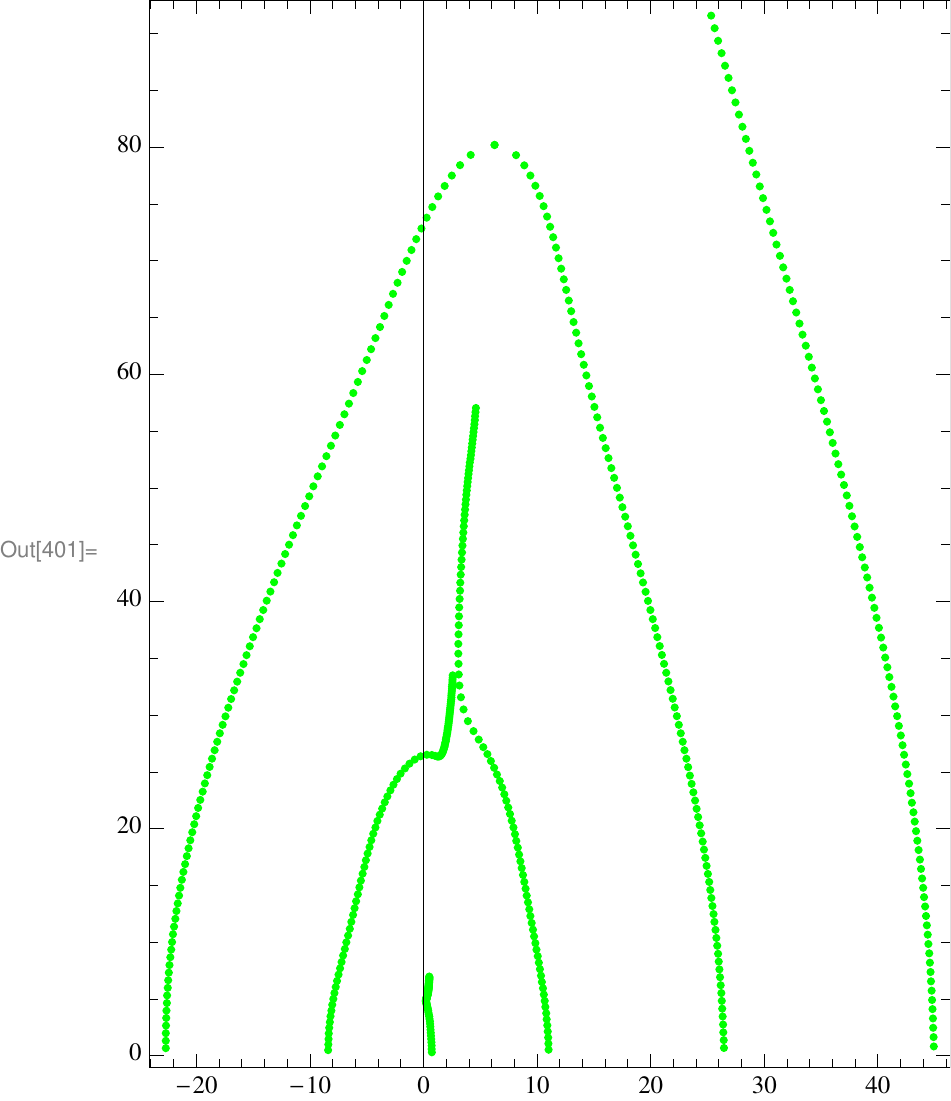}
\end{center}

\vskip 0.5cm
\begin{center}
\includegraphics[scale=0.23]{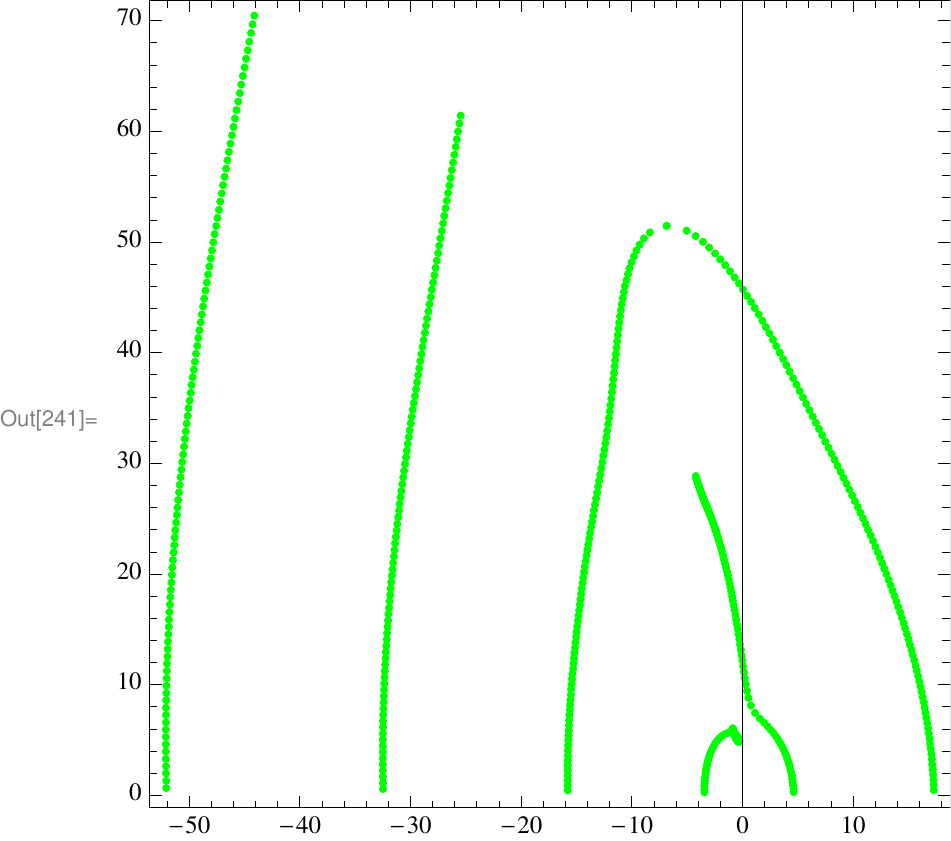}
\includegraphics[scale=0.25]{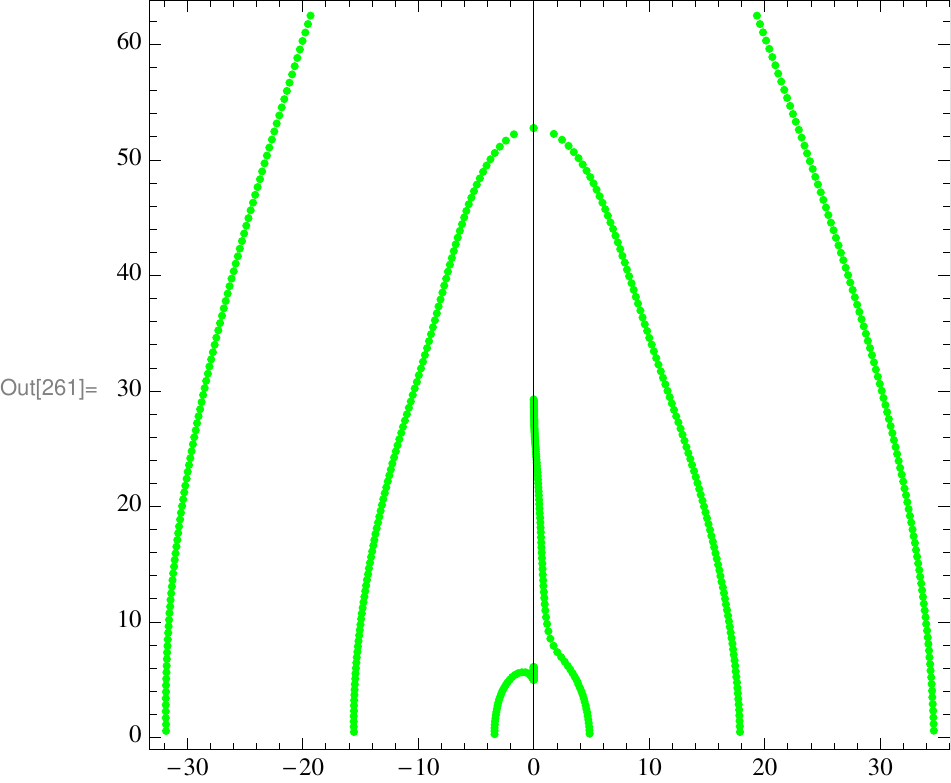}
\includegraphics[scale=0.23]{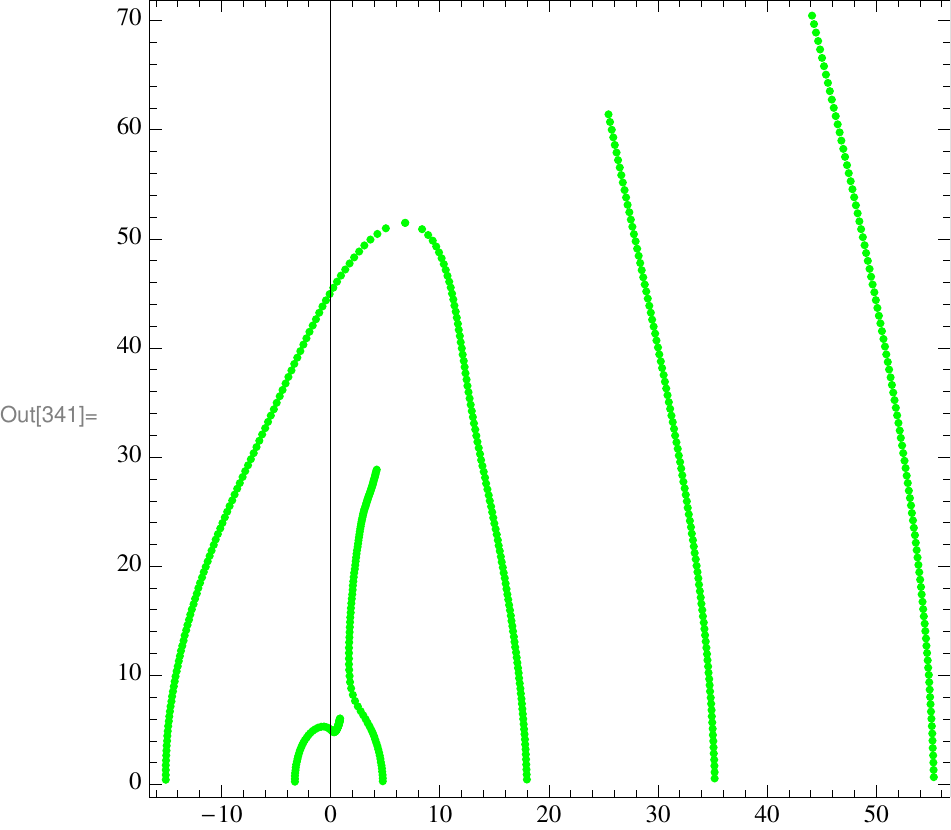}
\end{center}

\vskip 0.5cm
\begin{center}
\includegraphics[scale=0.24]{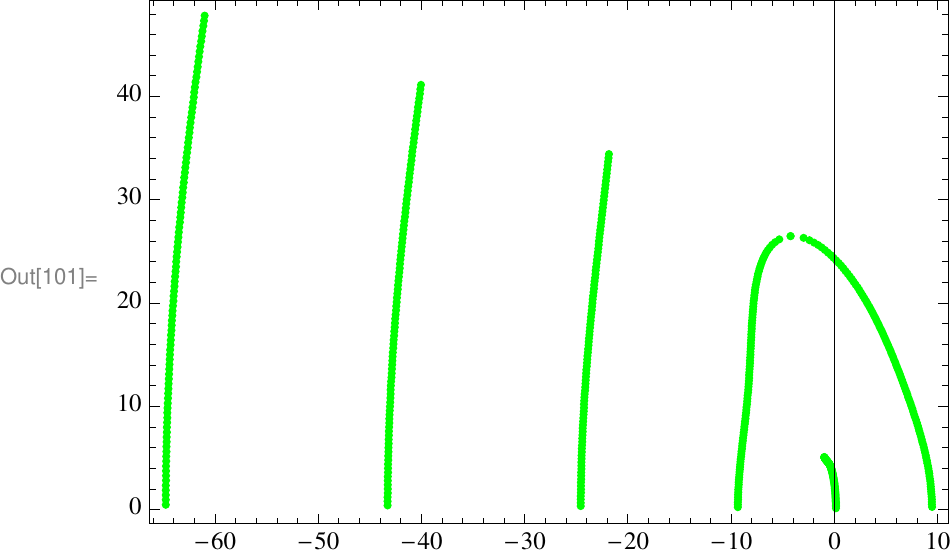}
\includegraphics[scale=0.25]{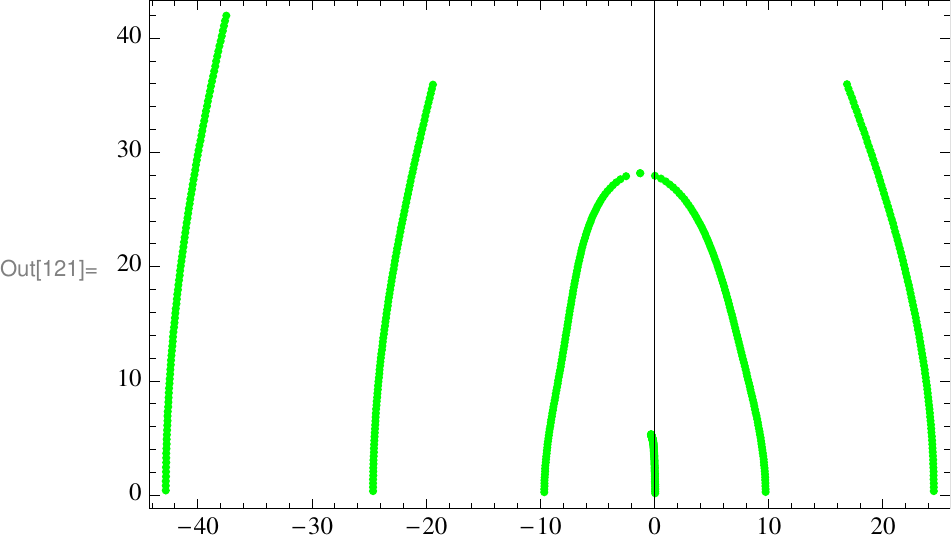}
\includegraphics[scale=0.25]{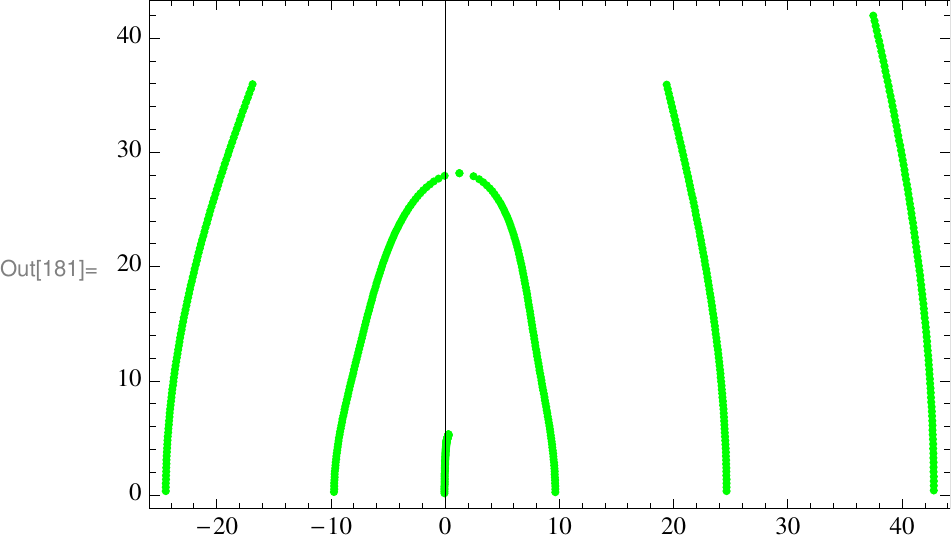}
\includegraphics[scale=0.24]{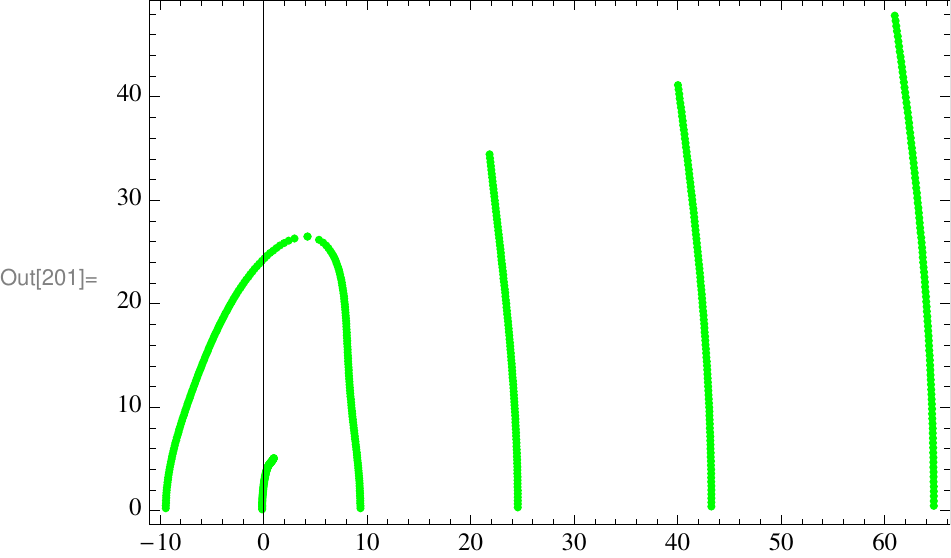}
\end{center}

\vskip 0.5cm
\begin{center}
\includegraphics[scale=0.25]{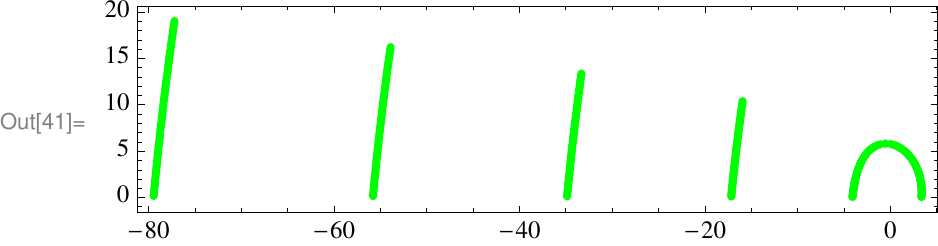}
\includegraphics[scale=0.25]{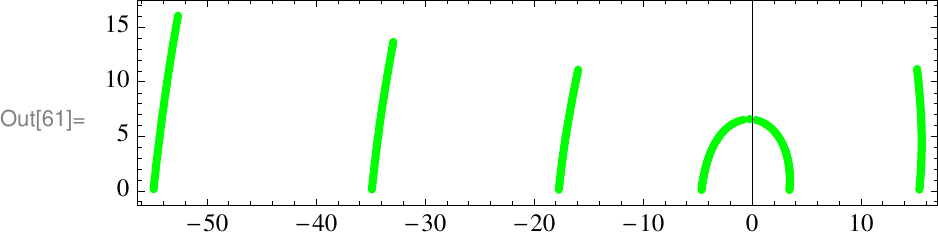}
\includegraphics[scale=0.25]{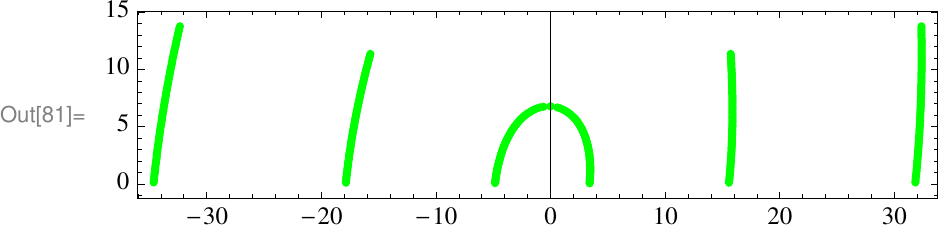}
\includegraphics[scale=0.25]{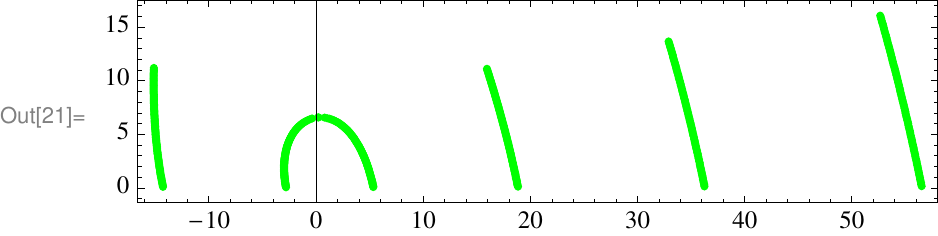}
\includegraphics[scale=0.25]{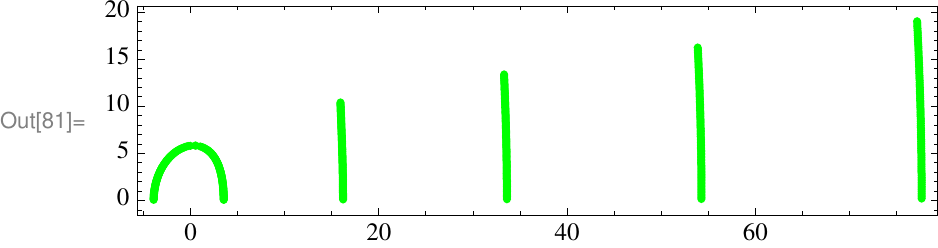}
\end{center}

\vskip 0.5cm
\caption{Monodromy for $m=5$.}
\label{fig4}
\end{figure}

\subsubsection{Monodromy of the algebraic spectrum of QES-sextic oscillator.}  By part (i) of the above conjecture about the structure of $\Si_m$, its level crossings points in the upper half plane are organized in $m$ ``horizontal" rows, where the  first row (i.e. the one closest to the real axis) contains $m$ points, the second row contains $m-1$ points, \dots, the $m$-th row contains one point, see Fig.~\ref{figromb}. 

If we order the level crossing points in the first row from left to right, i.e. according to increase of their real parts, then the corresponding monodromy operators look as follows. If we denote the left-most level crossing point in the first row by $b_1=u_1+i v_1$, then the permutation of the spectrum of $M_m(u_1)$ (which consists of $m+1$ real and distinct points) obtained, when $b$ traverses the loop $\ga_1$, is the simple transposition $(m,m+1)$. In other words,  two rightmost points of the spectrum of $M_m(u_1)$ will change places  when $b$ traverses  $\ga_1$. The  permutation corresponding to the second level crossing point in the first row is the simple transposition $(m-1,m)$.    The  permutation corresponding to the third level crossing point in the first row is the simple transposition $(m-2,m-1)$ etc. The case, $m=5$ is shown in the last row of Fig.~\ref{fig4}. 

Similarly, if we order level crossing points in the second row from left to right, i.e., according to increase of their real parts, then the corresponding monodromy operators look as follows. The  permutation corresponding to the left-most level crossing point in the second row is the  transposition $(m-1,m+1)$. The  permutation corresponding to the second level crossing point in the second row is the simple transposition $(m-2,m)$.    The  permutation corresponding to the third level crossing point in the first row is the simple transposition $(m-3,m-1)$ etc. The case, $m=5$ is shown in the fourth row of Fig.~\ref{fig4}. In  the third row we transpose pairs of eigenvalues separated by two intermediate eigenvalues etc. Finally, the transposition corresponding to the only level crossing point on the top is $(1,m+1)$, see the first row of  Fig.~\ref{fig4}.  

In other words, level crossing points of $M_m(b)$ is the upper half plane are in $1-1$-correspondence with all transpositions in the symmetric group on $m+1$ elements; those in the first row corresponding to simple transpositions, those in the second row corresponding to transpositions of pairs of elements which are separated by one intermediate element etc.  

\begin{remark} Pictures in Fig.~\ref{fig4}   show the trajectories of the eigenvalues of $M_m(b)$, when $b$ runs vertically from $u_j\in \bR$ to $b_j=u_j+iv_j$, $b_j$ being some level crossing point.  When $b=u_j$ all the eigenvalues of $M_m(b)$ are real.  When $b$ moves vertically up, the eigenvalues also move in the complex plane, and when $b$ reaches $b_j$ some of the eigenvalues collide.  One can trace back which initial 
eigenvalues collided and knowing that obtain the respective monodromy permutation. 
\end{remark}

\subsection{QES-sextic and quartic potential} This material is borrowed from an unpublished preprint \cite{ErGaIr}.  We will present a rescaling of the sequence of QES sextics which converges to the classical quartic oscillator. Recall that the latter oscillator corresponds to the Schr\"odinger equation
\begin{equation}\label{eq:quartic}
-y^{\prime\prime}(z)+(2z^4+\be z^2)y=\mu y
\end{equation}
with the initial conditions $y(\pm \infty)=0$ on the real axis. 

The spectrum of the classical quartic oscillator was studied in numerous papers since the early days of  quantum mechanics, see especially \cite{BW, Si, EG, Sh, Vo} and references therein.

To approximate the quartic potential by a sequence of QES-sextic potentials, set $n = 4m +  3$. Then the quasi-exactly solvable equation \eqref{eq:Sch} is related to
\begin{equation}\label{eq:last}
-y^{\prime\prime}(z) + [a^2z^6 + 2a\al z^4 + (\al^2 - an)z^2]y(z) = \mu y(z) 
\end{equation}
by the scaling $x = a^{1/4}z, \al = a^{1/2}b, \la = a^{1/2}\mu$. To approximate the quartic potential $2z^4 + \be z^2$ by the rescaled quasi-exactly solvable sextic potentials in \eqref{eq:last} as $m\to \infty$, let $\al = n^{1/3}(1 + sn^{-2/3}), 
a = n^{-1/3}(1 + tn^{-2/3})$. Then $b = \al/a^{1/2} = n^{1/2}(1 + (s - t/2)n^{-2/3} + O(n^{-4/3}).$ Substituting expression for $a$ and $\al$ into \eqref{eq:last}, we get the potential
$$n^{-2/3}(1+O(n^{-2/3)})z^6 +2(1+(s+t)n^{-2/3} +stn^{-4/3})z^4 +((2s-t)+s^2n^{-2/3})z^2.$$ 
Hence $\be = 2s- t = 2(n^{-1/2}b - 1)n^{2/3} + O(n^{-2/3}).$

Figure~\ref{FigS7} shows location of the level crossing points of $\la(b)$ for the rescaled
sextic, and Figure~\ref{fig10}, which is taken from \cite{DePh} shows the same for the quartic oscillator. (See also Fig.~1 in \cite{Sh}.)

\begin{figure}

\begin{center}
\includegraphics[scale=0.3]{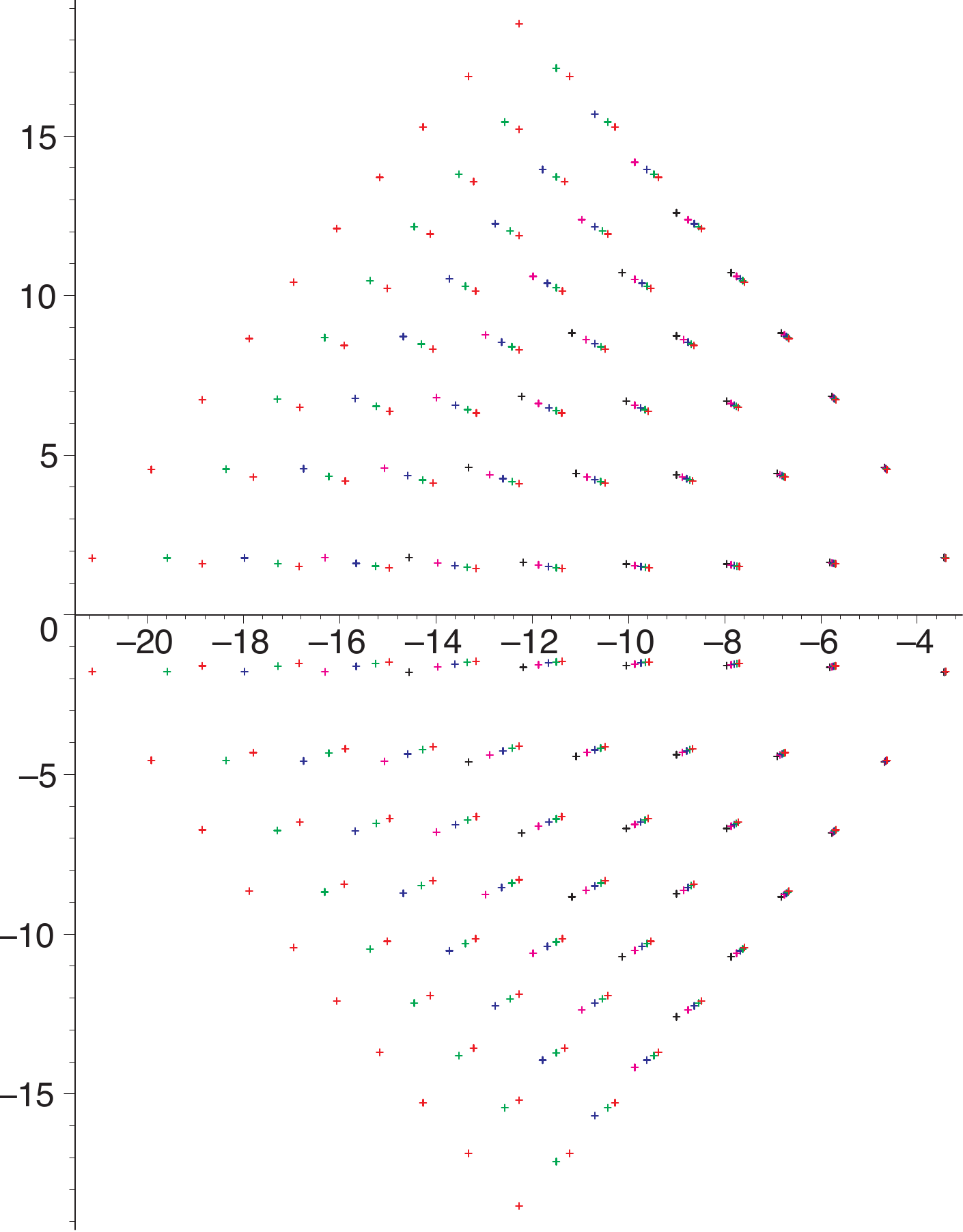} 
\end{center}


\caption{Level crossing points for rescaled QES sextics with $m = 6, 7, 8, 9, 10$.}
\label{FigS7}
\end{figure}

\begin{figure}

\begin{center}
\includegraphics[scale=0.35]{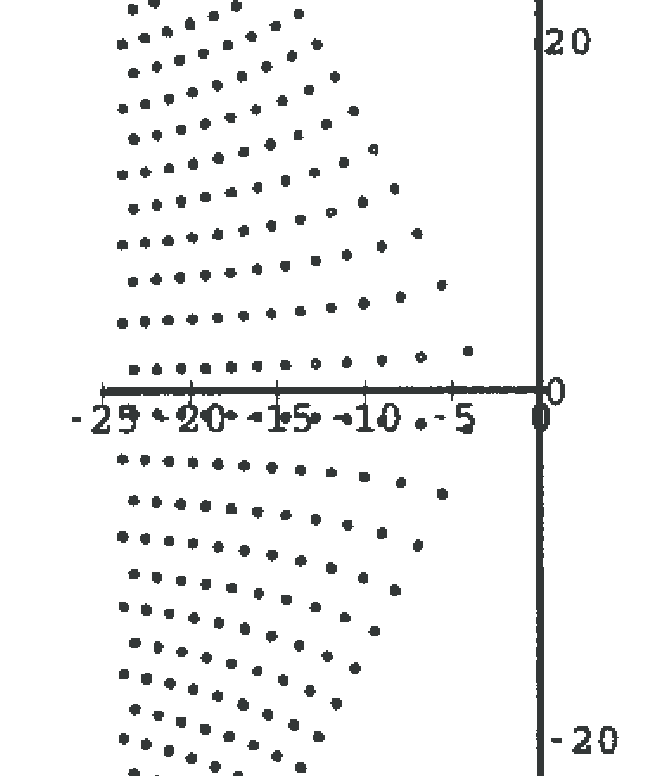} 
\end{center}


\caption{Level crossing points of the quartic oscillator, see \cite{DePh}.}
\label{fig10}
\end{figure}

\medskip\noindent
\subsubsection{Monodromy  of the full spectrum of the quartic oscillator.}  The level crossing set of the quartic oscillator naturally splits into ``horizontal" rows, see Fig.~\ref{fig10}.  Let us order 
level crossing points in each ``horizontal" row from right to left, i.e., in the order of decrease of their real parts. To each level crossing point in the upper half plane of Fig.~\ref{fig10}, we associate a loop similar to $\ga_j$'s above. Namely, we start from the point on the real axis with the same real part as the level crossing point under consideration, move straight up almost reaching the level crossing point, traverse a complete circle around it counterclockwise and go vertically down to the starting point on the real axis.  If we now denote by $b_{k,l}=u_{k,l}+iv_{k,l}$ the $k$-th level crossing point in the $l$-th ``horizontal" row, then the monodromy operator acting on the spectrum of \eqref{eq:quartic} with $\be=u_{k,l}$ is given by the transposition $(k,k+l)$. Here the latter spectrum is real, simple, countable and bounded from below. Its points are naturally labelled by positive integers in the order of increase. 

One can find a previous numerical study of the set of level points of the quartic oscillator and an attempt to determine its monodromy in \cite{Sh}. 

\section{Final Remarks}

\begin{problem} Find the linear differential operator of the minimal  order such that  the latter $\C_b(z)$  is its solution. (The existence of such an operator is guaranteed by the fact that $\C_b(z)$ is a Nilsson-class function, \cite {Nil}.) 
\end{problem}


\begin{thebibliography}{30}

\bibitem{BD} C.~Bender, G.~Dunne, {Quasi-exactly solvable systems and orthogonal polynomials}. J. Math. Phys. 37, (1996),  6--11. 

\bibitem{BDM} C.~Bender, G.~Dunne, M.~Moshe, {Semiclassical analysis of quasi-exact solvability}.  Phys. Rev A  55(2), (1997), 2625--2629. 

 
 \bibitem {BoSh} R.~B\o gvad, B.~Shapiro, 
On mother body measures with algebraic Cauchy transform, L'Enseignement Math., to appear. 


\bibitem{BW} C.~Bender, T.~Wu, {Anharmonic oscillator}. Phys. Rev. (2) 184, (1969),  1231--1260. 








\bibitem{DePh} E.~Delabaere, F.~Pham, Resurgent methods in semi-classical asymptotics, Annales de l'Inst. Poincar\'e, sect. A, t. 71 (1999) 1--94.

\bibitem{Em} A.~Emch, {On plane algebraic curves with a given system of foci}, Bull. AMS, vol 25 (1918), 157--161.

\bibitem {EG} A.~Eremenko and A.~Gabrielov, Analytic continuation of eigenfunctions of a quartic oscillator, Comm. math. phys., 287 (2009) 431--457.

\bibitem{ErGaIr} A.~Eremenko, A.~Gabrielov, Irreducibility of some spectral determinants, arXiv: 0904.1714.

\bibitem{KvA}
A.~B.~J.~Kuijlaars, W.~Van Assche, {\em The asymptotic zero distribution of orthogonal polynomials with varying recurrence coefficients,} J. Approx. Theory {\bf 99} (1999), 167--197.

\bibitem{Nil} N.~Nilsson, {\em Some growth and ramification properties of certain integrals on algebraic manifold},. Ark. Mat. 5 1965 463--476 (1965). 

\bibitem{Sa} G.~Salmon, A treatise on Higher Plane Curves. Intended as a Sequel to a Treatise on Conic Sections, Hodges and Smith, Third Edition, 1879.  

\bibitem{Si} B.~Simon, Coupling Constant Analyticity for the Anharmonic Oscillator, Ann. Phys., vol. 58 (1970), 76-136. 

\bibitem {Sh} P.~Shanley, {Spectral properties of the scaled quartic anharmonic oscillator}, Ann. Phys. 186 (1988), 292--324. 

\bibitem {ShTa1} B.~Shapiro, and M.~Tater,  Asymptotics of spectral polynomials, Acta Polytechnica vol 47(2-3) (2007) 32--35. 



\bibitem{ShTaQu} B.~Shapiro, M.~Tater, On spectral asymptotics of quasi-exactly solvable quartic and Yablonskii-Vorob'ev polynomials, arXiv:1412.3026, submitted.

\bibitem{SBD} V.Singh, S.N. Biswas and K. Datta, Anharmonic oscillator and the analytic theory of continued fractions, 
 Phys. Rev. D18(1978),
1901--1908.



\bibitem {Tu} A.~Turbiner, {Quasi-exactly solvable problems and $\mathfrak{sl}(2)$ algebra}, Comm. Math. Phys. 118 (1988). 467--474. 

\bibitem {TuUsh} A.~Turbiner, A,~Ushveridze, {Spectral singularities and the quasi exactly solvable problem}, Phys. Lett. 126 A (1987), 181--183. 

\bibitem {Ush} A.~Ushveridze, {Quasi-exactly solvable models in quantum mechanics}, Sov. J. Part.Nucl. 20 (1989). 504--528. 

\bibitem {UshB} A.~Ushveridze, {\em Quasi-exactly solvable models in quantum mechanics.}  Institute of Physics Publishing, Bristol, 1994. xiv+465 pp.


\bibitem {Vo} A.~Voros,  {The return of the quartic oscillator: the complex WKB method.} Ann. Inst. H. Poincar\'e Sect. A (N.S.) 39 (1983), no. 3, 211--338.






\end{thebibliography}
\end{document}